\documentclass[journal,onecolumn,11pt]{IEEEtran}
\usepackage{amsfonts,amsmath,amssymb}
\usepackage{indentfirst, setspace}
\usepackage{url,float}
\usepackage{color}
\usepackage[a4paper,ignoreall]{geometry}
\usepackage{longtable}
\usepackage{graphicx}
\usepackage[a4paper,ignoreall]{geometry}
\usepackage{multicol}
\usepackage{stfloats}
\usepackage{enumerate}
\usepackage{cite}
\usepackage{amsthm}
\usepackage{booktabs}
\usepackage{mathrsfs}
\usepackage{multirow}
\usepackage{amssymb}
\usepackage{verbatim}
\usepackage{cases}
\usepackage[square, comma, sort&compress, numbers]{natbib}
\usepackage[overload]{empheq}
\def\qu#1 {\fbox {\footnote {\ }}\ \footnotetext { From Qu: {\color{red}#1}}}
\def\hqu#1 {}
\def\kq#1 {\fbox {\footnote {\ }}\ \footnotetext { From KangQuan: {\color{blue}#1}}}
\def\hkq#1 {}

\usepackage{stfloats}

\usepackage{array}
\usepackage{soul}
\usepackage{booktabs}

\newtheorem{Th}{Theorem}

\newtheorem{Lemma}[Th]{Lemma}
\newtheorem{Def}[Th]{Definition}

\newcommand{\tr}{{\rm Tr}}

\newcommand{\gf}{{\mathbb F}}

\makeatletter
\newcommand{\figcaption}{\def\@captype{figure}\caption}
\newcommand{\tabcaption}{\def\@captype{table}\caption}
\makeatother

\begin{document}
	\title{Two new infinite classes of APN functions}
	\author{{ Kangquan Li, Yue Zhou, Chunlei Li and Longjiang Qu}
		\thanks{\noindent Kangquan Li, Yue Zhou and Longjiang Qu are with the College of Liberal Arts and Sciences,
			National University of Defense Technology, Changsha, 410073, China and Hunan Engineering Research Center of Commercial Cryptography Theory and Technology Innovation.
			Chunlei Li is with the Department of Informatics, University of Bergen, Bergen N-5020, Norway.
			The work of Longjiang Qu was supported by the Nature Science Foundation of China (NSFC) under Grant (No.~62032009),  National Key R$\&$D Program of China (No.~2017YFB0802000). The work of Yue Zhou was supported by the Training Program for Excellent Young  Innovators of Changsha (No.\ kq1905052). The work of  Chunlei Li was supported by the Research Council of Norway (No.~247742/O70 and No.~311646/O70) and the National Natural Science Foundation of China under Grant (No.~61771021).
			\smallskip		
		\textbf{Emails}: 	likangquan11@nudt.edu.cn, gabelozhou@gmail.com, chunlei.li@uib.no and ljqu\_happy@hotmail.com.
		}
	}
	\maketitle{}

\begin{abstract}
	In this paper, we present two new infinite classes of APN functions over $\gf_{{2^{2m}}}$ and $\gf_{{2^{3m}}}$, respectively. The first one is with bivariate form and obtained by adding special terms, $\sum(a_ix^{2^i}y^{2^i},b_ix^{2^i}y^{2^i})$, to a known class of APN functions by {G{\"{o}}lo{\v{g}}lu} over $\gf_{{2^m}}^2$. The second one is of the form $L(z)^{2^m+1}+vz^{2^m+1}$ over $\gf_{{2^{3m}}}$, which is a generalization of one family of APN functions by Bracken et al. [Cryptogr. Commun. 3 (1): 43-53, 2011]. 
	The calculation of the CCZ-invariants $\Gamma$-ranks of our APN classes over $\gf_{{2^8}}$ or $\gf_{{2^9}}$
	indicates that  they are CCZ-inequivalent to all known infinite families of APN functions. Moreover,  by using the code isomorphism, we see that our first APN family covers an APN function over $\gf_{{2^8}}$ obtained through the switching method by Edel and Pott in [Adv. Math. Commun. 3 (1): 59-81, 2009]. 
\end{abstract}


\section{Introduction}

S-boxes are crucial nonlinear components in block ciphers. In order to resist some known attacks, S-boxes used in block ciphers are required to satisfy a variety of cryptographic criteria, such as low differential uniformity for the differential attacks \cite{biham1991differential}. The definition of the differential uniformity is as follows.
\begin{Def}
	\cite{nyberg1993differentially}
	Let $f$ be a function over $\gf_{{2^n}}$. Then the differential uniformity of $f$ is 
	$$\delta_f = \max_{a\in\gf_{{2^n}}^{*}, b\in\gf_{{2^n}}} \#\{ z: z \in \gf_{{2^n}} | f(z+a) + f(z) = b  \}. $$ 
\end{Def}
 \noindent It is well known that for even characteristic, the almost perfect nonlinear (APN for short) functions \cite{nyberg1993differentially} with differential uniformity $2$ provide the best resistance to the differential attacks. Moreover, for a quadratic function $f$ over $\gf_{{2^n}}$, $f$ is APN if and only if the equation
 \begin{equation*}
 	\label{APN}
 	f(z+a)+f(z)+f(a) = 0
 \end{equation*}
has exactly two solutions $z=0, a$ in $\gf_{{2^n}}$ for any $a\in\gf_{{2^n}}^{*}$. The construction of new infinite classes of APN functions has been one of the most important topics {in the study of cryptographic functions}. For detailed information about APN functions, we invite the readers to consult {Carlet's recent book} \cite[Ch. 11]{carlet-2020}.

{By far known APN functions are mostly represented in  univariate and bivariate forms. While APN functions with bivariate forms can be also represented in univariate forms
according to certain isomorphism between $\gf_{2^{2m}}$ and $\gf_{2^m}^2$, their corresponding univariate forms are usually more complex.}
In order to present a concise form of APN families, in the following we summarize all known APN families from their initial expressions. 
 In this paper, we use $z$ and $x,y$ to denote the variables of univariate and bivariate forms, respectively. 
APN monomials are classical functions with univariate form.  All known APN monomials can be found in Table \ref{APN monomials}, {which was conjectured to be complete \cite{Dobbertin99-Niho}}. For all known infinite classes of non-monomials APN functions with univariate form, please {refer to Table \ref{APN polynomials}, in which $\tr_{m}^{n}$ denotes the trace function from $\gf_{{2^n}}$ to $\gf_{2^m}$ for any $m \mid n$, i.e., $\tr_m^n(z) = z+z^{2^m}+\cdots+z^{2^{\left(\frac{n}{m}-1\right)m}}$, and $\tr^n_{1}$ is shortly denoted as $\tr_{n}$.}
  Bivariate construction is a powerful method to get APN functions over $\gf_{2^m}^2$. Carlet \cite{Carlet11}, Zhou and Pott \cite{Zhou2013}, Taniguchi \cite{Taniguchi2019}, Calderini, Budaghyan and Carlet \cite{Calderini2020} constructed different APN functions over $\gf_{2^{m}}^2$ of the form $f(x,y)=(xy,G(x,y))$ by choosing  distinct bivariate functions $G(x,y)$. Very recently, G{\"o}lo{\v{g}}lu considered the construction of APN functions using bi-projective polynomials and obtained two classes of APN functions. All APN families of bivariate forms are listed in Table \ref{APN bivariate}, where $P_1(z)= z^{2^k+1}+az+b$ and $P_2(z) = (cz^{2^i+1}+bz^{2^i}+1)^{2^{m/2}+1}+z^{2^{m/2}+1}$. 
\begin{table}[!t]
	\caption{All Known APN monomials over $\gf_{2^n}$} \label{APN monomials}
	\centering
	\small
	\begin{tabular}{c c cc }	
		\toprule
		Family&	Function  & Conditions & Ref. \\
		\midrule
		Gold & $z^{2^i+1}$ & $\gcd(i,n)=1$ &  \cite{Gold1968} \\
		Kasami & $z^{2^{2i}-2^i+1}$ & $\gcd(i,n)=1$ &  \cite{Kasami1971The} \\
		Welch & $z^{2^t+3}$ & $n=2t+1$ &  \cite{Dobbertin99} \\
		Niho-1& $z^{2^t+2^{t/2}-1}$ & $n=2t+1, t$ even &  \cite{Dobbertin99-Niho} \\
		Niho-2& $z^{2^t+2^{(3t+1)/2}-1}$ & $n=2t+1, t$ odd &  \cite{Dobbertin99-Niho} \\
		Inverse& $z^{2^{2t}-1} $ & $n=2t+1$ &  \cite{nyberg1993differentially} \\
		Dobbertin  & $z^{2^{4i}+2^{3i}+2^{2i}+2^i-1}$ & $n=5i$ & 
		\cite{Dobbertin2001} \\
		\bottomrule
	\end{tabular}
\end{table}

\begin{table}[!htbp]
		\caption{All Known APN infinite families with univariate forms (non-monomials) over $\gf_{2^n}$} \label{APN polynomials}
\centering
\small
	\begin{tabular}{c c cc }	
		\toprule
		No.&	Function  & Conditions & Ref. \\
		\midrule
		\begin{tabular}[c]{@{}l@{}} F1- F2\end{tabular} & $z^{2^s+1}+u^{2^k-1}z^{2^{ik}+2^{mk+s}}$                                               & \begin{tabular}[c]{@{}l@{}} $n=pk, \gcd(k,3)=\gcd(s,3k)=1,$\\ $p\in\{3,4\}, i = sk \pmod p, m = p-i,$\\ $n\ge 12$, $u$ primitive in $\gf_{2^n}^{*}$\end{tabular} & \cite{Budaghyan2008-1}
	\\
\hline
		F3                                              & \begin{tabular}[c]{@{}l@{}} $sz^{2^i(q+1)}+z^{2^i+1}+z^{q(2^i+1)}$\\ $+cz^{2^iq+1}+c^qz^{2^i+q}+z^{q+1}$\end{tabular} & \begin{tabular}[c]{@{}l@{}} $q=2^m$, $n=2m$, $\gcd(i,m)=1,$\\ $c\in\gf_{2^n}, s\in\gf_{2^n}\backslash\gf_{2^m}, z^{2^i+1}+cz^{2^i}+$\\ $c^qz+1$ has no solution $x$ with $x^{q+1}=1$\end{tabular}  & \cite{Budaghyan2008}  \\ 
	\hline  \specialrule{0em}{1pt}{1pt}
 F4 
	& $z^3+a^{-1}\tr_n(a^3z^9)$                                                       & $a\neq0$                                                                           & \cite{Budaghyan2009} \\
	\hline \specialrule{0em}{1pt}{1pt}
	F5	                                     &  $z^3+a^{-1}\tr_3^n(a^3z^9+a^6z^{18})$                  & $3 \mid n, a\neq0$                                                                            &  \cite{budaghyan2009on} \\
	\hline \specialrule{0em}{1pt}{1pt}
	F6                                              & $z^3+a^{-1}\tr_3^n(a^6z^{18}+a^{12}z^{36})$                                                    & $3 \mid n, a\neq0$                                                                                                             &  \cite{budaghyan2009on} \\
	\hline
	\begin{tabular}[c]{@{}l@{}}F7-F9\end{tabular}   & \begin{tabular}[c]{@{}l@{}} $uz^{2^s+1}+u^{2^m}z^{2^{-m}+2^{m+s}}+$\\ $vz^{2^{-m}+1}+wu^{2^m+1}z^{2^s+2^{m+s}}$\end{tabular} & \begin{tabular}[c]{@{}l@{}} $n=3m, \gcd(m,3)=\gcd(s,3m)=1, v,w\in\gf_{2^m}$ \\ $vw\neq1, 3 \mid {m+s}, u$ primitive in $\gf_{2^n}^{*}$ \end{tabular}               &  \cite{bracken2011a} \\
	\hline
	F10                                             & \begin{tabular}[c]{@{}l@{}} $a^2z^{2^{2m+1}+1}+b^2z^{2^{m+1}+1}+$ \\ $az^{2^{2m}+2}+bz^{2^m+2}+(c^2+c)z^3$ \end{tabular} & \begin{tabular}[c]{@{}l@{}} $n=3m, m$ odd, $L(z)=az^{2^{2m}}+bz^{2^m}+cz$\\ satisfies the conditions of Lemma 8 of \cite{budaghyan2020constructing} \end{tabular}                   & \cite{budaghyan2020constructing}  \\
	\hline
		F11                                             & \begin{tabular}[c]{@{}l@{}}$z^3+w z^{2^{i+1}}+w^2z^{3\cdot 2^m}$\\ $+z^{2^{i+m}+2^m}$ \end{tabular}  & \begin{tabular}[c]{@{}l@{}} $n=2m, m$ odd, $3\nmid m$, $w$ primitive  \\ in $\gf_{2^2}, s = m-2, (m-2)^{-1} \pmod n$ \end{tabular}        &  \cite{budaghyan2020a}\\
		\hline
		F12                                               & \begin{tabular}[c]{@{}l@{}} $a\tr_m^n(bz^3)+a^q\tr_m^n(b^3z^9)$\\ \end{tabular}   & \begin{tabular}[c]{@{}l@{}}$n=2m, m $ odd, $q=2^m$, $a\notin \gf_q,$\\ $b$ not a cube\end{tabular}                   & \cite{Zheng2021} \\
		\bottomrule
	\end{tabular}
\end{table}

\begin{table}[!htbp]
	\caption{All known APN families with bivariate forms over $\gf_{2^m}^2$} \label{APN bivariate}
	\centering
	\small
	\begin{tabular}{c c cc }	
		\toprule
		No. &	Function  & Conditions & Ref. \\
		\midrule
		F13 & $(xy, x^{2^k+1} + \alpha x^{(2^k+1)2^i})$ & $\gcd(k,m)=1$, $m$ even, $\alpha$ non-cubic &  \cite{Zhou2013} \\
		F14 & $(xy, x^{2^{3k}+2^{2k}}+ax^{2^{2k}}y^{2^k}+by^{2^k+1})$ & $\gcd(k,m)=1$, $P_1$ no root in $\gf_{2^{m}}$ &  \cite{Taniguchi2019} \\
		F15 & $(xy, x^{2^i+1}+x^{2^{i+m/2}}y^{2^{m/2}}+bxy^{2^i}+cy^{2^i+1})$ & $m$ even, $\gcd(i,m)=1$, $P_2$ no root in $\gf_{2^{m}}$ &  \cite{Calderini2020} \\
		F16 & $(x^{2^i+1}+xy^{2^i}+y^{2^i+1}, x^{2^{2i}+1}+x^{2^{2i}}y+y^{2^{2i}+1})$ & $\gcd(3i,m)=1$ &\cite{Gologlu2020}\\
		F17 & $(x^{2^i+1}+xy^{2^i}+y^{2^i+1}, x^{2^{3i}}y+xy^{2^{3i}})$ & $\gcd(3i,m)=1$, $m$ odd & \cite{Gologlu2020} \\
		\bottomrule
	\end{tabular}
\end{table}

In this paper, we construct two new infinite families of APN functions.
One class is inspired by the method of Dillon \cite{Dillon2006} and its generalization of Budaghyan and Carlet \cite{Budaghyan2008}.
In \cite{Dillon2006}, Dillon presented a way to construct APN functions of the form 
\begin{equation}
	\label{Dillon}
	f(z)=z(Az^2+B^q+Cz^{2q})+z^2(Dz^q+Ez^{2q})+Gz^{3q}
\end{equation}
 over $\gf_{q^2}$ with $q=2^m$. Particularly, in \cite{Budaghyan2008}, Budaghyan and Carlet obtained an infinite family of APN hexanomials based on construction \eqref{Dillon}, i.e., F3 in Table \ref{APN polynomials}. Let $f_1(z) = Az^3+Cz^{2q+1}+Dz^{q+2}+Gz^{3q}$. The coefficients for APNness of $f_1$ have been determined completely by Li, Li, Helleseth and Qu \cite{Li2020}. Later Chase and Lison\v{e}k \cite{Chase2020} proved that $f_1$ is APN if and only if $f_1$ is CCZ-equivalent to the Gold functions. Thus the APN hexanomials (F3 in Table \ref{APN polynomials}) obtained by Budaghyan and Carlet can be seen as  the summation of a known APN function and  $Bz^{q+1}+Ez^{2(q+1)}$. Inspired these observations, a natural idea is to try to find new APN functions over $\gf_{q^2}$ by adding special terms of the form $\sum c_iz^{2^i(q+1)}$, $c_i\in\gf_{q^2}^{*}$ to known APN families. Note that for the bivariate form over $\gf_q^2$, $\sum  c_iz^{2^i(q+1)}$ is actually of the form $\sum (a_ix^{2^i}y^{2^i},b_ix^{2^i}y^{2^i})$, $a_i,b_i\in\gf_q^{*}$. {Finally, we find a new infinite family of APN functions over $\gf_{2^{m}}^2$  by adding terms to F16 in Table \ref{APN bivariate} as follows
\begin{equation}\label{Eq-APN1}
f(x,y)=\left( x^3+xy^2+y^3+xy, x^5+x^4y+y^5+xy+x^2y^2 \right),
\end{equation}
 where $\gcd(3, m)=1$.
It is interesting that by the code isomorphism, our APN family includes an APN function over $\gf_{{2^8}}$ obtained through the switching method by Edel and Pott \cite{EdelP09}.}
%\begin{Th}
%	\label{newAPNtheorem}
%	Let $m$ be a positive integer with $\gcd(3,m)=1$. Let  Then $f(x,y)$ is APN over $\gf_{2^m}^2$. 
%\end{Th}
Another class is motivated from the following APN quadrinomial over $\gf_{{2^{3m}}}$ obtained by Bracken et al. \cite{bracken2011a} $$f(z)=uz^{2^s+1}+u^{2^m}z^{2^{-m}+2^{m+s}}+vz^{2^{-m}+1}+wu^{2^m+1}z^{2^s+2^{m+s}}.$$
We assume that $w\neq0$. Choose $\gamma\in\gf_{{2^m}}$ satisfying $\gamma^{1-2^s}=w$, which always exists since $\gcd\left(2^m-1,2^s-1\right)=2^{\gcd(m,s)}-1=1$. Then $f(\gamma z)^{2^m} = \gamma^{2^s+1} (L(z)^{2^m+1}+(vw+1)z^{2^m+1})$, where $L(z)=u^{2^m}z^{2^{m+s}}+z$. Thus $f$ is linear equivalent to an APN family of the form $L(z)^{2^m+1}+vz^{2^m+1}$ with $L$ a permutation and $v\neq0$. By choosing a linearized permutation trinomial $L$, we propose another new infinite family of APN functions over $\gf_{{2^{3m}}}$ as follows: 
{
\begin{equation}\label{Eq-APN2}
f(z) = L(z)^{2^m+1}+vz^{2^m+1},
\end{equation}
where $\gcd(s,m)=1$, $v\in\gf_{2^{m}}^{*}$, $\mu\in\gf_{2^{3m}}$ with $\mu^{2^{2m}+2^m+1}\neq1$ and $ L(z) = z^{2^{m+s}}+\mu z^{2^s}+z$ permuting $\gf_{2^{3m}}$.
}

The remainder of this paper is organized as follows. We first in Section \ref{CCZ} discuss the issue of CCZ-equivalence that researchers may care about the most. Sections \ref{new APN} and \ref{newAPN2} prove the APNness of the functions in \eqref{Eq-APN1} and \eqref{Eq-APN2},  respectively. Section \ref{conclusion} concludes the work of this paper and presents some related problems. 

\section{CCZ-equivalence}
\label{CCZ}
{
Two functions $f$ and $g$ over $\gf_{2^n}$ are said to be \emph{Carlet-Charpin-Zinoviev (CCZ) equivalent} if there is an affine permutation of $\gf_{2^n}^2$ that maps the graph $G_f = \left\{ (z,f(z)): z \in \gf_{2^n} \right\}$ to the graph $G_g = \left\{ (z,g(z)): z \in \gf_{2^n} \right\}$. 
It is known that the CCZ equivalence 
the most general equivalence relation that preserves the differential uniformity of a function over $\gf_{2^n}$.
When an infinite class of APN functions is newly constructed, it is of great interest to investigate whether it is essentially new, i.e.,  it is CCZ-inequivalent to all known infinite classes of APN functions.
On the other hand, it is very difficult to theoretically prove two APN functions are CCZ-inequivalent. 
Therefore, a common practice of examining the CCZ-equivalence of two functions 
$f$ and $g$ is to evaluate the code isomorphism for small parameters. 
To be more concrete, for a function over $\gf_{{2^n}}$, an associated linear code $\mathcal{C}_f$ is defined by the following generating matrix }
$$\mathcal{C}_f = \begin{pmatrix}
1 & 1 & \cdots & 1 \\
0 & u & \cdots & u^{2^n-1} \\
f(0) & f(u) & \cdots & f\left( u^{2^n-1} \right)
\end{pmatrix}$$
where $u$ is a primitive element of $\gf_{{2^n}}$. Two functions $f$ and $g$ are CCZ-equivalent if and only if $\mathcal{C}_f$ and $\mathcal{C}_g$ are isomorphic. 
Another common method is to compute the CCZ-invariants, i.e., properties that remain invariant under CCZ-equivalence of two functions in a fixed finite field. {If two functions in question exhibit different values and/or properties regarding certain CCZ-invariant,  then they must be CCZ-inequivalent. One CCZ-invariant is the $\Gamma$-rank \cite{EdelP09}, which is defined as the rank of the incidence matrix of a design $\mathrm{dev}(G_f)$ for a function $f$ over $\gf_{2^n}$, whose set of points is $\gf_{2^n}^2$ and set of blocks is $\left\{ (z+a,f(z)+b): z \in\gf_{2^n} \right\}$ for $a,b\in\gf_{2^n}$. }

In this section, we explain our newly constructed APN families are CCZ-inequivalent to all known APN classes by comparing the $\Gamma$-ranks of the representatives from all known APN families and ours over $\gf_{2^8}$ and $\gf_{{2^9}}$, see Tables \ref{gamma-rank1} and \ref{gamma-rank2}\footnote{ In Tables \ref{gamma-rank1} and \ref{gamma-rank2}, the $\Gamma$-ranks of the representatives from all known APN families are retrieved from the website https://boolean.h.uib.no/mediawiki/index.php/Tables}. 
In Table \ref{gamma-rank1} (resp. Table \ref{gamma-rank2}), $u$ and $v$ denote any primitive element of $\gf_{2^8}$ (resp. $\gf_{{2^9}}$)  and $\gf_{2^4}$, respectively. The ``Ref." denotes  the corresponding known APN families in Tables \ref{APN monomials}, \ref{APN polynomials} and \ref{APN bivariate}. The different values of $\Gamma$-ranks show that our newly found APN infinite  class is new.  Another fact worth mentioning is that by the code isomorphism,  our first APN family in \eqref{Eq-APN1} covers one (No. 1.9) of  APN functions over $\gf_{{2^8}}$ in \cite[Table 9]{EdelP09}, which were obtained through the switching method by Edel and Pott. The covered APN function is
$$p(z) = z^3+ u \tr_8\left( u^{63} z^3+ u^{252}z^9 \right) +u^{154} \tr_8\left( u^{68}z^3+u^{235} z^{9}\right) + u^{35} \tr_8 \left( u^{216}z^3+u^{116} z^{9}\right),$$ 
where $u$ is a primitive element of $\gf_{{2^8}}.$

\begin{table*}[!htbp]
	\caption{ CCZ-inequivalent representatives from the known APN families and ours over $\gf_{2^8}$ and their $\Gamma$-ranks  } \label{gamma-rank1}
	\centering
	\begin{tabular}{c c cc }	
		\toprule
		No.&	Function  & $\Gamma$-rank & Ref. \\
		\midrule
		1 & $z^3$ & 11818 & Gold \\
		2& $z^9$ & 12370 & Gold \\
		3& $z^{57}$ & 15358 & Kasami \\
		4& $z^3+z^{17}+u^{48}z^{18}+z^3x^{33}+uz^{34}+z^{48}$ & 13200 & F3 \\
		5& $z^3+\tr_8(z^9)$ & 13800 & F4 \\
		6& $z^3+u^{-1}\tr_8(u^3z^9) $ & 13842 & F4 \\
		7& $ (xy, x^3+vy^{12}) $ & 13642 & F13 \\
		8& $ (xy, x^{12}+x^4y^2+y^3) $ & 13700 & F14 \\
		9& $(xy, x^{12}+x^4y^2+v^7y^3)$ & 13798 & F14 \\
		10& $(x^3+xy^2+y^3, x^5+x^4y+y^5)$  & 13642&  F16 \\
		11& $(xy, x^3+x^2y+vx^4y^8+v^5y^3)$ & 13960 &  F15 \\
		12& $(x^3+xy^2+y^3+xy, x^5+x^4y+y^5+xy+x^2y^2)$ & 14034 & this paper  \\
		\bottomrule
	\end{tabular}
\end{table*}

\begin{table*}[!htbp]
	\caption{ CCZ-inequivalent representatives from the known APN families and ours over $\gf_{2^9}$ and their $\Gamma$-ranks  } \label{gamma-rank2}
	\centering
	\begin{tabular}{c c cc }	
		\toprule
		No.&	Function  & $\Gamma$-rank & Ref. \\
		\midrule
		1 & $z^3$ & 38470 & Gold \\
		2& $z^5$ & 41494 & Gold \\	 
		3& $z^{17}$ & 38470 & Gold \\
		4& $z^{13}$ & 58676 & Kasami \\
		5& $z^{241}$ & 61726 & Kasami \\
		6& $z^{19}$ & 60894 & Welch \\
		7& $z^{255}$ & 130816 & Inverse \\
		8& $z^3+\tr_9(z^9)$ & 47890 & F4 \\
		9 & $z^3+\tr_3^9(z^9+z^{18}) $ & 48428 & F5 \\
		10 & $z^3+\tr_3^9(z^{18}+z^{36}) $ & 48460 & F5 \\
		11& $z^3+u^{246}z^{10}+u^{47}z^{17}+u^{181}z^{66}+u^{428}z^{129}$ & 48596 & F10 \\
		12& $(z^{16}+u^5z^2+z)^9+u^{73}z^9$ & 48558 & this paper\\ 
		%	13 & $(z^{16}+u^{13}z^2+z)^9+u^{73}z^9$ & 48564 & this paper\\
		\bottomrule
	\end{tabular}
\end{table*}

\section{A new infinite class of APN functions over $\gf_{{2^m}}^2$}

\label{new APN}

In this section, we will {prove that the function in \eqref{Eq-APN1} is APN}. Before that, we first give several useful lemmas. The following one is to determine the number of solutions of cubic equations over $\gf_{{2^m}}$.
\begin{Lemma}
	\cite{williams1975note}
	\label{cubic}
	Let $a,b\in\gf_{2^m}^{*}$ and define 
	$$f(z)=z^3+az+b, h(t)=t^2+bt+a^3.$$ 
	Let $t_1, t_2$ be two solutions of $h(t)$. Then 
	\begin{itemize}
		\item $f$ has three zeros in $\gf_{2^m}$ if and only if $\tr_m\left(\frac{a^3}{b^2}\right)=\tr_m(1)$, $t_1$ and $t_2$ are cubes in $\gf_{2^m}$ (resp. $\gf_{2^{2m}}$) when $m$ is even (resp. odd).
		\item $f$ has exactly one zero in $\gf_{2^m}$ if and only if $\tr_m\left(\frac{a^3}{b^2}\right)\neq\tr_m(1)$.
		\item $f$ has no zeros in $\gf_{2^m}$ if and only if $\tr_m\left(\frac{a^3}{b^2}\right)=\tr_m(1)$, $t_1$ and $t_2$ are not cubes in $\gf_{2^m}$ (resp. $\gf_{2^{2m}}$) when $m$ is even (resp. odd).
	\end{itemize}
\end{Lemma}

The following two lemmas are very important for our proof.

\begin{Lemma}
	\label{cube}
	Let $\gcd(m,3)=1$ and $\omega\in\gf_{2^2}\backslash\gf_2$. Then $\omega$ is not cubic in $\gf_{2^m}$ (resp. $\gf_{2^{2m}}$) if $m$ is even (resp. odd).
\end{Lemma}
\begin{proof}
	We only consider the case where $m$ is even. It suffices to show that the equation $z^3=\omega$ has no solution in $\gf_{{2^m}}$. If not, then $\omega^{\frac{2^m-1}{3}}=1,$ which means that $3\mid {\frac{2^m-1}{3}}$, i.e., $9 \mid {2^m-1}$. It is a contradiction with the fact $\gcd(m,3)=1$.
\end{proof}

%\begin{Lemma}
%	\label{helleseth2010}
%	\cite{helleseth2010}
%	Let $\alpha\in\gf_{2^n}^{*}$ and $\gcd(j,n)=1$. Then $x^{2^j+1}+x+\alpha=0$ has no zeros in $\gf_{2^n}$ if and only if $D_n(\alpha)\neq0$.  
%\end{Lemma}

\begin{Lemma}
	\label{Z^{2^j+1}+Z+1=0}
	Let $\gcd(3,m)=1$. Then $$z^3+z+1=0$$ has no solution in $\gf_{2^m}$. 
\end{Lemma}
\begin{proof}
	If there exists some element $z\in\gf_{{2^m}}$ such that $z^3+z+1=0$, then $z\in\gf_{2^3}$. Since $\gcd(3,m)=1$, $z\in\gf_{2^{\gcd(3,m)}}=\gf_2$. However, it is clear that for $z=0$ or $1$, $z^3+z+1\neq0.$
\end{proof}

Since  the resultant of polynomials will be used in our proof, we now recall some basic facts about the resultant of two polynomials. Given two non-zero polynomials of degrees $n$ and $m$ respectively
$$ u(x) = a_mx^m+a_{m-1}x^{m-1}+\cdots+a_0$$ and $$ v(x) = b_nx^n+b_{n-1}x^{n-1}+\cdots+b_0 $$
with $a_m\neq0, b_n\neq0$ and coefficients in a field or in an integral domain ${R}$, their resultant $\mathrm{Res}(u,v)\in {R}$ is the determinant of the following matrix:
$$ \small \begin{pmatrix} 
	a_m & a_{m-1} & \cdots & \cdots & \cdots & \cdots & a_0 & 0 & 0 & 0 \\
	0 & a_m & a_{m-1} & \cdots & \cdots & \cdots & \cdots &  a_0 & 0 & 0 \\
	&  & \ddots & \ddots & & & &  &   \ddots &      \\
	0 & 0 & 0 & a_m & a_{m-1} & \cdots & \cdots & \cdots & \cdots & a_0 \\
	b_n & b_{n-1} & \cdots & \cdots & b_0 & 0 & 0 & 0 & 0 & 0\\
	& \ddots & \ddots &  & & \ddots & &  &   &      \\ 
	0   & 0 & 0 & 0 & 0 &  b_n & b_{n-1} &  \cdots &   \cdots &  b_0    \\
\end{pmatrix}.
$$
For a field $K$ and two polynomials $F(x,y), G(x,y) \in K[x,y]$, we use $ \mathrm{Res}(F,G,y)$ to denote the resultant  of $F$ and $G$ with respect to $y$. It is the resultant of $F$ and $G$ when considered as polynomials in the single variable $y$. In this case, $ \mathrm{Res}(F,G,y)\in K[x]$ belongs to the ideal generated by $F$ and $G$, and thus any $a,b$ satisfying $F(a,b)=0$ and $G(a,b)=0$ is such that $\mathrm{Res}(F,G,y)(a)=0$ (see \cite{LN1997}). 

The main result in this section is as follows.

%Now we give the proof of Theorem \ref{newAPNtheorem}. We recall the new infinite family of APN functions firstly. \mkq{It's not proper to have another same theorem here. Recalling the function here is sufficient.} 

\begin{Th}
	Let $m$ be a positive integer with $\gcd(3,m)=1$ and $$f(x,y)=\left( x^3+xy^2+y^3+xy, x^5+x^4y+y^5+xy+x^2y^2 \right).$$ Then $f(x,y)$ is APN over $\gf_{2^m}^2$. 
\end{Th}

\begin{proof}
	It suffices to show that for any $(a,b)\neq(0,0) \in \gf_{2^m}^2$, the equation 
	\begin{equation}
	\label{Eq_APN}
	f(x+a,y+b)+f(x,y)+f(a,b)=0
	\end{equation} 
	has exactly two solutions $(x,y)=(0,0), (a,b)$ in $\gf_{2^m}^2.$ By {a simple calculation}, Eq. \eqref{Eq_APN} is equivalent to the following equation system 
	\begin{subequations} 
		\renewcommand\theequation{\theparentequation.\arabic{equation}}  
		\label{eq_1}  
		\begin{empheq}[left={\empheqlbrace\,}]{align}
		& ax^2+\left(a^2+b^2+b\right)x + (a+b)y^2 + \left(a+b^2\right) y = 0 \label{eq_1_1} \\ 
		& (a+b)x^4+b^2x^2+\left(a^4+b\right)x + by^4+a^2y^2+ \left( a^4+a+b^4\right) y = 0.  \label{eq_1_2} 
		\end{empheq}
	\end{subequations}

	Let $$F(x,y,a,b)= ax^2+\left(a^2+b^2+b\right)x + (a+b)y^2 + \left(a+b^2\right) y$$ and $$G(x,y,a,b) = (a+b)x^4+b^2x^2+\left(a^4+b\right)x + by^4+a^2y^2+ \left( a^4+a+b^4\right) y.$$ 
	{Then the resultant of $F$ and $G$ aiming at $y$ equals } \begin{equation}
	\label{resultant}
	\mathrm{Res}(F,G,y) = \left(a^3+ab^2+b^3\right)^2x(x+a)H(x,a,b)H(x+a,a,b),
	\end{equation}
	where 
	$$H(x,a,b) = x^3 + (a^2 + ab + a + b^2 + b + 1)x + a^3 + a^2b + a + b^3 + b^2 + 1.$$
	First of all, we have $a^3+ab^2+b^3\neq 0$ for any $(a,b)\neq(0,0) \in \gf_{2^m}^2$. If not, then for some element  $(a,b)\neq(0,0) \in \gf_{2^m}^2$, $a^3+ab^2+b^3=0$. If $b=0$, then the above equation becomes $a^3=0$, which is a contradiction with the fact $(a,b)\neq(0,0)$. If $b\neq0$, then we have $c^3+c+1=0$, where $c=\frac{a}{b}\in\gf_{2^m}$, which is {in contradiction} with Lemma \ref{Z^{2^j+1}+Z+1=0}. In addition, according to Eqs. \eqref{eq_1}, i.e., $F(x,y,a,b)=G(x,y,a,b)=0$, we have $x=0, a$ or $H(x,a,b)=0$. Now we consider the equation 
	\begin{equation}
	\label{eq_3}
	x^3 + A x + B=0,
	\end{equation} 
	where $A = a^2 + ab + a + b^2 + b + 1$ and $B= a^3 + a^2b + a + b^3 + b^2 + 1.$ In the following, we prove that $B=0$ if and only if $(a,b)=(1,1)$. Plugging $a=a_1+b$ into $B=0$ and { simplifying it}, we get
	\begin{equation}
	\label{eq_4}
	a_1^3+(b^2+1)a_1+(b+1)^3=0.
	\end{equation}
	If $b=1,$ then $a_1=0$ and thus $a=a_1+b=1$. If $b\neq1$, plugging $a_1=(b+1)a_2$ into Eq. \eqref{eq_4} and {simplifying it}, we have $a_2^3+a_2+1=0$, which has no solution in $\gf_{2^n}$ from Lemma \ref{Z^{2^j+1}+Z+1=0}. Thus $B=0$ if and only if $(a,b)=(1,1)$. 
	In the following, we divide our proof into two cases: $(a,b)=(1,1)$ and $(a,b)\neq(1,1)$. 
	
	{\bfseries Case 1.} If $(a,b)=(1,1)$, then Eqs. \eqref{eq_1} becomes 
	\begin{subequations} 
		\renewcommand\theequation{\theparentequation.\arabic{equation}}  
		\label{eq_11}  
		\begin{empheq}[left={\empheqlbrace\,}]{align}
		& x^2+x =0 \label{eq_11_1} \\ 
		& x^2+y^4+y^2+y = 0.  \label{eq_11_2} 
		\end{empheq}
	\end{subequations} 
	From Eq. \eqref{eq_11_1}, we known $x=0$ or $1$. If $x=0$, plugging it into Eq. \eqref{eq_11_2}, we have $y^4+y^2+y=0$ and then $y=0$ by Lemma \ref{Z^{2^j+1}+Z+1=0}. If $x=1$, together with Eq. \eqref{eq_11_2}, we get $y^4+y^2+y+1=0$, which means $y=1$. Thus in this case, Eqs. \eqref{eq_1} has two solutions $(x,y)=(0,0)$ and $(1,1)$ in $\gf_{2^m}^2$.
	
	{\bfseries Case 2.} If $(a,b)\neq (1,1)$, then $B\neq0$. Let $h(t) = t^2+Bt+A^3$. By computation, we have 
	$$\frac{A^3}{B^2} = \frac{C}{B}+\frac{C^2}{B^2}+1,$$
	where $$C = a^2b+a^2+ab^2+a+b^2+b.$$
	Thus $\tr_m\left( \frac{A^3}{B^2} +1\right)=0$ and the equation $h(t)=0$ has two solutions $t_1= C+\omega B$ and  $t_2= C+\omega^2 B$ in $\gf_{2^m}$ (resp. $\gf_{2^{2m}}$) if $m$ is even (resp. odd), where $\omega\in\gf_{2^2}\backslash\gf_2$.  Moreover, $t_1=\omega(\omega^2C+B)  = \omega \left(a+\omega b +\omega^2\right)^3$, which is not cubic by Lemma \ref{cube}. Thus from Lemma \ref{cubic}, the equation $x^3+ Ax+ B=0$ has no solution in $\gf_{2^m}$. Hence from Eq. \eqref{resultant}, we have $x=0$ or $a$ and then $y=0$ or $b$, respectively.
	
	All in all, Eqs. \eqref{eq_1} has exactly two solutions $(x,y)=(0,0), (a,b)$ in $\gf_{2^m}^2$ for any $(a,b)\neq (0,0)\in \gf_{2^m}^2$. Therefore, $f$ is APN over $\gf_{2^m}^2.$
\end{proof}

\section{A new infinite class of APN functions over $\gf_{{2^{3m}}}$}
\label{newAPN2}

In this section, we will show the univariate function defined in \eqref{Eq-APN2} is APN, namely the following theorem.

\begin{Th}
	\label{newAPNFq3}
	Let  $\gcd(s,m)=1$ and $v\in\gf_{2^{m}}^{*}$. Choose $\mu\in\gf_{2^{3m}}$ such that $\mu^{2^{2m}+2^m+1}\neq1$ and $ L(z) = z^{2^{m+s}}+\mu z^{2^s}+z$ permutes $\gf_{2^{3m}}$.  Then $f(z) = L(z)^{2^m+1}+vz^{2^m+1}$ is APN over $\gf_{2^{3m}}$.
\end{Th}

 Before that, we first prove some important lemmas.  

\begin{Lemma}
	\label{LL}
	Let $\gcd(m,s)=1$, $\mu\in\gf_{2^{3m}}$ satisfy $\mu^{2^{2m}+2^m+1}\neq1$ and $L_{\beta}(z) = z^{2^{m+s}}+\mu z^{2^s}+\beta z$ with $\beta \in\gf_{2^{m}}$. If $L_1$ permutes $\gf_{2^{3m}}$, then so does $L_{\beta}$ for any $\beta\in\gf_{2^{m}}$.
\end{Lemma}
\begin{proof}
	It suffices to show that $L_{\beta}(a)\neq0$ for any $a\in\gf_{2^{3m}}^{*}$. Otherwise,
{taking  $\epsilon\in\gf_{2^{m}}$, replacing $a$ by $\epsilon a$ in $L_{\beta}(a)=0$ and simplifying it}, we get $a^{2^{m+s}}+\mu a^{2^s}+\beta\epsilon^{1-2^s} a =0$. If $\beta =0$, we have  $a^{2^{m+s}}+\mu^{2^m} a^{2^s}=0$, which is {in contradiction} with the condition $\mu^{2^{2m}+2^m+1}\neq1$. If $\beta\neq0$, let $\epsilon\in\gf_{2^{m}}^{*}$ satisfy $\epsilon^{2^s-1} = \beta $, which always exists since $\gcd\left(2^m-1,2^s-1\right)=2^{\gcd(m,s)}-1=1$. Then the above equation becomes $L_1(a)=0$ for some {$a\in\gf_{2^{3m}}^{*}$}, which is {in contradiction} with the condition that $L_1$ permutes {$a\in\gf_{2^{3m}}^{*}$}. 
\end{proof}
\begin{Lemma}
	\label{LL1}
	Let $\mu\in\gf_{2^{3m}}$, $ L(z) = z^{2^{m+s}}+\mu z^{2^s}+z $ and $ L^{'}(z) = z^{2^{m+s}}+ \mu^{2^m} z^{2^m}+z $. Then $ L $ permutes $ \gf_{2^{3m}} $ if and only if so does $L^{'}$.
\end{Lemma}
\begin{proof}
	For the linear polynomial $ L $, we denote by $ L^{*} $ its {adjoint polynomial \cite{PottPMB18}}, i.e., $L^{*}(z)=z^{2^{2m-s}}+\mu^{2^{3m-s}}z^{2^{3m-s}}+z$. It is well known that $L$ permutes $\gf_{2^{3m}}$ if and only if so does $L^{*}$. Moreover, it is easy to check that $L^{'}(z) = (L^{*}(z))^{2^{m+s}}$. Thus  $L$ permutes $\gf_{2^{3m}}$ if and only if so does $L^{'}$.
\end{proof}

Let $\gcd(s,m)=1$, $a\in\gf_{2^{3m}}$ and $v\in\gf_{2^{m}}^{*}$. Choose $\mu\in\gf_{2^{3m}}^{*}$ such that $\mu^{2^{2m}+2^m+1}\neq1$ and {$L(z) = z^{2^{m+s}}+\mu z^{2^s}+z$ }permutes $\gf_{2^{3m}}$. 
Define 
%\begin{equation}
%\label{ABCDE}	
%\left\{
%\begin{array}{lr}
%A = ( a^{2^{m+s}} + \mu a^{2^s} +a ) a^{2^{2m+s}}  \\ 
%B = ( \mu^{2^m+1}a^{2^s} + \mu^{2^m} a + a^{2^{2m+s}} + a^{2^m} ) a^{2^{m+s}} \\
%C =  ( a^{2^{m+s}} + \mu a^{2^s} + (v+1) a ) a^{2^m} \\
%D = ( \mu a^{2^{2m+s}} + \mu^{2^m+1} a^{2^{m+s}} + \mu a^{2^m}  )a^{2^s} \\
%E = ( a^{2^{2m+s}} + \mu^{2^m} a^{2^{m+s}} + (v+1) a^{2^m}) a.
%\end{array}
%\right.
%\end{equation} 
	\begin{equation}
\label{ABCDE-new}	
\left\{
\begin{array}{lll}
A = L(a) a^{2^{2m+s}},  \\
B = ( L(a)^{2^m}+\mu^{2^m}L(a)) a^{2^{m+s}}  \\
C =  (L(a)+va) a^{2^m},  \\
D = \mu L(a)^{2^m}a^{2^s}, \\
E = (L(a)+va)^{2^m}a.
\end{array}
\right.
\end{equation} 
and denote
\begin{equation}
\label{U1U2U3U4}	
\left\{
\begin{array}{lr}
U_1 = D^{2^{2m}}E^{2^m+1} + AC^{2^{2m}}E^{2^m} +B^{2^m} C^{2^{2m}+1}  \\ 
U_2 = A^{2^{2m}}E^{2^m+1} + BC^{2^{2m}}E^{2^m} + C^{2^{2m}+1}D^{2^m} \\
U_3 = B^{2^{2m}}E^{2^m+1} + C^{2^{2m}}DE^{2^m} + A^{2^m}C^{2^{2m}+1} \\
U_4 = C^{2^{2m}+2^m+1}+E^{2^{2m}+2^m+1}.
\end{array}
\right.
\end{equation} 

\begin{equation}
\label{V1V2V3V4}	
\left\{
\begin{array}{lr}
V_1 = A^{2^{2m}+2}C^{2^m} + ABC^{2^m}D^{2^{2m}} + AB^{2^m+1}E^{2^{2m}} + A^2D^{2^m}E^{2^{2m}}  \\ 
V_2 = A^{2^{2m}+2}E^{2^m} + ABD^{2^{2m}}E^{2^m} + A^{2^{2m}+1}B^{2^m}C + ACD^{2^{2m}+2^m}   \\
V_3 = A^{2^{2m}+1}B^{2^m}E + AB^{2^m+1}C^{2^{2m}} +A^2C^{2^{2m}}D^{2^m} + AD^{2^{2m}+2^m}E \\
V_4 = (B^{2^m+1}+AD^{2^m})(AB^{2^{2m}}+D^{2^{2m}+1})+(A^{2^{2m}+1}+BD^{2^{2m}})(A^{2^m+1}+B^{2^m}D).
\end{array}
\right.
\end{equation}

The following lemma is crucial to the proof in this section. 
\begin{Lemma}
	\label{key-lemma}
	Let $A,B,C,D,E$ be defined as in \eqref{ABCDE-new}, $U_i, V_i$ with $i=1,2,3,4$ be defined as in \eqref{U1U2U3U4} and \eqref{V1V2V3V4}, respectively. Then for any $a\in\gf_{2^{3m}}^{*}$,
	\begin{enumerate}[(i)]
		\item $A+B+C+D+E=0,$	$ABCDE\neq0$ and $C+E\neq0$;
		\item $U_i V_i\neq0$ with $i=1,2,3$;
		\item $U_4=V_4=0$;
		\item  $U_2V_1^{2^{s}}+U_1V_2^{2^{s}}+U_3V_1^{2^{s}}+U_1V_3 ^{2^{s}}= 0;$
		\item $U_2V_1^{2^{s}}+U_1V_2^{2^{s}}\neq0$.
	\end{enumerate}
\end{Lemma}

\begin{proof}
%	The expressions of $A,B,C,D,E$ can be transformed as follows:

	(i) {It is readily seen that}
	$$
	A+B+C+D+E= L(a)\left(a^{2^{m+s}} +\mu a^{2^{s}}  + a\right)^{2^m} + L^{2^m}(a)\left( a^{2^{m+s}}  + \mu a^{2^s} + a\right) + a^{2^{m}+1}(v+v^{2^m}) = 0.
	$$	
	Next we show $ABCDE\neq0$. First, the fact $AD\neq0$ is clear since $A=L(a)a^{2^{2m+s}}$, $D=\mu L(a)^{2^m}a^{2^s}$ and $L$ permutes $\gf_{2^{3m}}$. Second, let $H(x)=x^{2^m}+\mu^{2^m} x$. Then $H$ permutes $\gf_{2^{3m}}$ since $\mu^{2^{2m}+2^m+1}\neq1$. In addition, it is easy to check that $B = H(L(a))a^{2^{m+s}}$ and thus $B\neq0$ for any $a\in\gf_{2^{3m}}^{*}$. Third, {$CE \neq 0$ is equivalent to $L(a)+va=a^{2^{m+s}}+\mu a^{2^s}+(v+1)a\neq0$, which can be easily verified} by Lemma \ref{LL}. 
	
	Finally,	if there exists some $a\in\gf_{2^{3m}}^{*}$ such that $C+E = (L(a)+va)a^{2^m}+(L(a)+va)^{2^m}a=0$, then $L(a)+va= \eta a$ for some $\eta\in\gf_{2^{m}}^{*}$, which is also impossible by Lemma \ref{LL}.  Thus $C+E\neq0$ for any $a\in\gf_{2^{3m}}^{*}$.
	
	(ii) Let $$U=\mu^{2^m}a^{2^{m+s}+1}+a^{2^{2m+s}+1}+a^{2^{m+s}+2^m}+\mu^{2^m+1}a^{2^{2m}+2^{m+s}}+\mu^{2^{2m}+1}a^{2^{2m+s}+2^m}+\mu a^{2^{2m+s}+2^{2m}}.$$
	Plugging the expressions of $A,B,C,D,E$ into that of $U_1,U_2,U_3$ and {investigating their factorizations}, we get 
	\begin{equation}
	\label{U1U2U3new}	
	\left\{
	\begin{array}{lr}
		U_1 = va^{2^{2m+s}+2^m}C^{2^{2m}}U^{2^{2m}}= va^{2^m}C^{2^{2m}}(a^{2^s}U)^{2^{2m}}  \\ 
		U_2 = va^{2^{m+s}+2^m}C^{2^{2m}}U^{2^m} =va^{2^m}C^{2^{2m}}(a^{2^s}U)^{2^m}\\
		U_3 = va^{2^m+2^s}C^{2^{2m}}U =va^{2^m}C^{2^{2m}}(a^{2^s}U). \\
	\end{array}
	\right.
\end{equation}
	Moreover, let $$V = a^{2^m+2^s}+\mu^{2^m}a^{2^{2m}+2^{m+s}}+\mu^{2^{2m}}a^{2^{2m+s}+2^m}+a^{2^{2m+s}+2^{2m}}$$ and 
	$$T=(\mu^{2^{2m}+2^m+1}+1)a^{2^s}+\mu^{2^{2m}+2^m}a+\mu^{2^{2m}}a^{2^m}+a^{2^{2m}}.$$ Similarly, plugging the expressions  of $A,B,C,D,E$ into that of $V_1,V_2,V_3$ and {investigating their factorizations}, we obtain 
	\begin{equation}
	\label{V1V2V3new}	
	\left\{
	\begin{array}{lr}
		V_1 = va^{2^{2m+s+1}+2^{m+s}+2^{2m}}L(a)TV^{2^{2m}}  = va^{2^{2m+s+1}+2^{m+s}}L(a)T(aV)^{2^{2m}}\\ 
		V_2 = va^{2^{2m+s+1}+2^{m+s}+2^m}L(a)TV^{2^m}  = va^{2^{2m+s+1}+2^{m+s}}L(a)T(aV)^{2^m} \\
		V_3 = va^{2^{2m+s+1}+2^{m+s}+1}L(a)TV= va^{2^{2m+s+1}+2^{m+s}}L(a)T(aV). \\
	\end{array}
	\right.
\end{equation}
	Thus in order to show that $U_iV_i\neq0$ with $i=1,2,3$, it suffices to prove that $UVT\neq0$. Let $P=a^{2^{2m+s}}+\mu^{2^m}a^{2^{m+s}}.$ Then $P\neq0$ for any $a\in\gf_{2^{3m}}^{*}$ since the binomial $z^{2^{2m+s}}+\mu^{2^m}z^{2^{m+s}}$ permutes $\gf_{2^{3m}}$ trivially. In addition, we have $$V=a^{2^{2m}}P+a^{2^m}P^{2^m}.$$ If there exists some $a\in\gf_{2^{3m}}^{*}$ such that $V=0$, then $P=\tau a^{2^m}$ for some $\tau\in\gf_{2^{m}}^{*}$. Furthermore, we have $a^{2^{2m+s}}+\mu^{2^m}a^{2^{m+s}}+\tau a^{2^m}=0,$ i.e., $a^{2^{m+s}}+\mu a^{2^s} + \tau a =0$, which is impossible by Lemma \ref{LL}. Thus $V\neq0$. Moreover, $U\neq0$ due to a {crucial observation} 
	$$U=V^{2^{2m}}+\mu V,$$ which is a permutation in $V$ over $\gf_{{2^{3m}}}$. 
	Finally, if there exists some $a\in\gf_{2^{3m}}^{*}$ such that $T=0$, i.e., 
	\begin{equation}
	\label{T=0}
	(\mu^{2^{2m}+2^m+1}+1)a^{2^s}+\mu^{2^{2m}+2^m}a+\mu^{2^{2m}}a^{2^m}+a^{2^{2m}} =0.
	\end{equation}
	 Raising \eqref{T=0} {to the $2^m$-th power, one gets}
	\begin{equation}
	\label{T=0-1}
	(\mu^{2^{2m}+2^m+1}+1)a^{2^{m+s}}+\mu^{2^{2m}+1}a^{2^m}+\mu a^{2^{2m}}+a =0.
	\end{equation}
	Comparing \eqref{T=0} and \eqref{T=0-1}, one can eliminate $\mu^{2^{2m}+1}a^{2^m}+\mu a^{2^{2m}}$ in \eqref{T=0-1} and obtains
	$$(\mu^{2^{2m}+2^m+1}+1) L(a) = 0,$$
	which is impossible since $\mu^{2^{2m}+2^m+1}\neq1$ and $L(a)\neq0$ for any $a\in\gf_{2^{3m}}^{*}$.

	(iii) It is trivial that $U_4=C^{2^{2m}+2^m+1}+E^{2^{2m}+2^m+1} =0$ due to the fact $E=C^{2^m}a^{1-2^{2m}}$.  The statement $V_4 =0$ holds due to the following four equations, which can be checked directly.
	\begin{align*}
	B^{2^{m}+1}+AD^{2^m} & = (L(a)^{2^m}+\mu^{2^m}L(a))^{2^m+1}a^{2^{2m+s}+2^{m+s}}+\mu^{2^m}L(a)^{2^{2m}+1}a^{2^{2m+s}+2^{m+s}}\\
	&=a^{2^{2m+s}+2^{m+s}}L(a)^{2^m}\left(L(a)^{2^{2m}}+\mu^{2^{2m}}L(a)^{2^m}+\mu^{2^{2m}+2^m}L(a)\right),
	\end{align*} 
	\begin{align*}
	AB^{2^{2m}}+D^{2^{2m}+1} &= L(a)a^{2^{2m+s}}(L(a)^{2^m}+\mu^{2^m}L(a))^{2^{2m}}a^{2^s}+\mu^{2^{2m}+1} L(a)^{2^m+1}a^{2^{2m+s}+2^s} \\
	& = a^{2^{2m+s}+2^s} L(a) \left(\mu L(a)^{2^{2m}} +\mu^{2^{2m}+1}L(a)^{2^m} +L(a) \right),
	\end{align*}
	\begin{align*}
	A^{2^{2m}+1}+BD^{2^{2m}} & = L(a)^{2^{2m}+1}a^{2^{2m+s}+2^{m+s}}+ (L(a)\mu^{2^m}+L(a)^{2^m})a^{2^{m+s}}\mu^{2^{2m}}L(a)a^{2^{2m+s}}\\
	&=a^{2^{2m+s}+2^{m+s}}L(a)\left(L(a)^{2^{2m}}+\mu^{2^{2m}}L(a)^{2^m}+\mu^{2^{2m}+2^m}L(a)\right)
	\end{align*} 
	and 
	\begin{align*}
{A^{2^{m}+1}+B^{2^m}D} &= L(a)^{2^m+1}a^{2^{2m+s}+2^s}+(L(a)\mu^{2^m}+L(a)^{2^m})^{2^m}a^{2^{2m+s}}\mu L(a)^{2^{m}}a^{2^s} \\
	& = a^{2^{2m+s}+2^s} L(a)^{2^m} \left(\mu L(a)^{2^{2m}} +\mu^{2^{2m}+1}L(a)^{2^m} +L(a) \right).
	\end{align*}

	(iv) Plugging \eqref{U1U2U3new} and \eqref{V1V2V3new} into $U_2V_1^{2^{s}}+U_1V_2^{2^{s}}+U_3V_1^{2^{s}}+U_1V_3 ^{2^{s}},$ we get 
	\begin{align*}
	& U_2V_1^{2^{s}}+U_1V_2^{2^{s}}+U_3V_1^{2^{s}}+U_1V_3 ^{2^{s}} \\
	=& \Delta\left(a^{2^{m+s}}U^{2^m}(a^{2^{2m}}V^{2^{2m}})^{2^s}+a^{2^{2m+s}}U^{2^{2m}}(a^{2^m}V^{2^m})^{2^s}+a^{2^s}U(a^{2^{2m}}V^{2^{2m}})^{2^s}+a^{2^{2m+s}}U^{2^{2m}}(aV)^{2^s}\right)\\
	=& a^{2^{2m+s}}\Delta \left( a^{2^{m+s}}U^{2^m}V^{2^{2m+s}}+a^{2^{m+s}}U^{2^{2m}}V^{2^{m+s}} + a^{2^s}UV^{2^{2m+s}} +a^{2^s}U^{2^{2m}}V^{2^s} \right),
	\end{align*}
	where {$\Delta = v^{2^s+1}a^{2^{2m+2s+1}+2^{m+2s}+2^m}C^{2^{2m}}L(a)^{2^s}T^{2^s}$}.
	 Moreover, it is easy to check $$a^{2^{2m+s}}U^{2^{2m}}+a^{2^{m+s}}U^{2^m}+a^{2^s}U=0$$ and 
	$$a^{2^{2m}}V^{2^{2m}}+a^{2^m}V^{2^m}+aV=0.$$ Thus we have 
	$$a^{2^{m+s}}U^{2^m}V^{2^{2m+s}}+a^{2^s}UV^{2^{2m+s}} = a^{2^{2m+s}}U^{2^{2m}}V^{2^{2m+s}}$$
	and 
	$$a^{2^{m+s}}U^{2^{2m}}V^{2^{m+s}} + a^{2^s}U^{2^{2m}}V^{2^s} = a^{2^{2m+s}}U^{2^{2m}}V^{2^{2m+s}}.$$
	Furthermore, we get 
	$$U_2V_1^{2^{s}}+U_1V_2^{2^{s}}+U_3V_1^{2^{s}}+U_1V_3 ^{2^{s}}=0.$$
	
	(v) By direct computations, we have 
	\begin{align*}
	&U_2V_1^{2^s} + U_1V_2^{2^s} \\
	=& \Gamma \left( a^{2^{m+s}}U^{2^m} (a^{2^{2m}}V^{2^{2m}})^{2^s} + a^{2^{2m+s}}U^{2^{2m}} (a^{2^m}V^{2^m})^{2^s} \right) \\
	=& \Gamma a^{2^{2m+s}+2^{m+s}}\left( UV^{2^{m+s}} +U^{2^m}V^{2^s} \right)^{2^m},
	\end{align*}
	where $\Gamma = v^{2^s+1}a^{2^{2m+2s+1}+2^{m+2s}+2^m}C^{2^{2m}}L(a)^{2^s}T^{2^s}$.
 Thus in order to prove $U_2V_1^{2^s} + U_1V_2^{2^s} \neq0$, it suffices to show $UV^{2^{m+s}} +U^{2^m}V^{2^s}\neq0$. If there exists some $a\in\gf_{2^{3m}}^{*}$ such that $UV^{2^{m+s}} +U^{2^m}V^{2^s}=0$, then $U=\gamma V^{2^s}$ for some $\gamma\in\gf_{2^{m}}^{*}$. In addition, since $U=V^{2^{2m}}+\mu V$, we obtain 
	$$V^{2^{2m}}+\mu V +\gamma V^{2^s} = 0.$$
	Let $\epsilon\in\gf_{2^{m}}$ satisfy $\epsilon^{2^s-1}=\gamma$. {Replacing $V$ with $\epsilon V$} in the above equation, we have $V^{2^{2m}}+\mu V +V^{2^s}=0,$ i.e., $$V^{2^{m+s}}+\mu^{2^m} V^{2^m} + V = 0,$$
	which is impossible by Lemma \ref{LL1} and the fact that $L$ permutes $\gf_{2^{3m}}$.
\end{proof}

Now we give the proof of Theorem \ref{newAPNFq3}.

\begin{proof}
	It suffices to show that for any $a\in\gf_{2^{3m}}^{*}$, the equation $f(az+a)+f(az)+f(a)=0$	has exactly two solutions $z = 0, 1$ in $\gf_{2^{3m}}$.
	{More specifically, we need to show the equation}
	\begin{equation}\label{main-Eq}\begin{split}
		& (L(az)+L(a))^{2^m+1}+  L(az)^{2^m+1}+ L(a)^{2^m+1}+{v(az+a)^{2^m+1} +v(az)^{2^m+1} +va^{2^m+1}}
		\\=& Az^{2^{2m+s}}+Bz^{2^{m+s}}+Cz^{2^m}+Dz^{2^s}+E{z} = 0
	\end{split}	
\end{equation} 
where $A,B,C,D,E$ are defined as in \eqref{ABCDE-new}, has exactly two solutions $z=0, \, 1$.
	By Lemma \ref{key-lemma} (i), we know $ABCDE\neq0$ for any $a\in\gf_{2^{3m}}^{*}$. 
	Raising \eqref{main-Eq} to its $2^m$-th power and its $2^{2m}$-th power, we have
	\begin{equation}
	\label{main-Eq1}
	A^{2^m}z^{2^{s}}+B^{2^m}z^{2^{2m+s}}+C^{2^m}z^{2^{2m}}+D^{2^m}z^{2^{m+s}}+E^{2^m}z^{2^m} = 0
	\end{equation}
	and 
	\begin{equation}
	\label{main-Eq2}
	A^{2^{2m}}z^{2^{m+s}}+B^{2^{2m}} z^{2^{s}}+C^{2^{2m}}z+D^{2^{2m}}z^{2^{2m+s}}+E^{2^{2m}}z^{2^{2m}} = 0,
	\end{equation}
	respectively. In the following, we will use the method of elimination twice and finally acquire two  equations \eqref{main-Eq4} and \eqref{main-Eq5}.
	After computing the summation of \eqref{main-Eq} multiplied by $E^{2^m}$ and \eqref{main-Eq1} multiplied by $C$ and simplifying it, we get
	\begin{equation}
	\label{main-Eq3}
	(AE^{2^m}+B^{2^m}C)z^{2^{2m+s}}+ C^{2^m+1} z^{2^{2m}} + ( BE^{2^m}+CD^{2^m})z^{2^{m+s}} + ( DE^{2^m} +A^{2^m}C )z^{2^s} + E^{2^m+1}z = 0.
	\end{equation} 
	Applying a similar operation on \eqref{main-Eq2} and \eqref{main-Eq3}, we have 
	$$	U_1z^{2^{2m+s}}+U_2z^{2^{m+s}}+U_3z^{2^s} + U_4z^{2^{2m}} =0,$$
	i.e.,
	\begin{equation}
	\label{main-Eq4}
	U_1z^{2^{2m+s}}+U_2z^{2^{m+s}}+U_3z^{2^s} =0,
	\end{equation}
	where $U_i$ with $i=1,2,3,4$ are defined as in \eqref{U1U2U3U4} and $U_4=0$ by Lemma \ref{key-lemma} (iii).  Similarly, by eliminating the terms $z^{2^{2m+s}}$ and $z^{2^{m+s}}$ from \eqref{main-Eq}, \eqref{main-Eq1} and \eqref{main-Eq2}, we obtain 
	$$V_1z^{2^{2m}}+V_2z^{2^{m}}+V_3z+V_4 z^{2^s} = 0,$$
	i.e., 
	\begin{equation}
	\label{main-Eq5}
	V_1^{2^s}z^{2^{2m+s}}+V_2^{2^s}z^{2^{m+s}}+V_3^{2^s}z^{2^s} =0,
	\end{equation}
	where $V_i$ with $i=1,2,3,4$ are defined as in \eqref{V1V2V3V4} and $V_4=0$ by Lemma \ref{key-lemma} (iii).
	
	{Now the summation} of \eqref{main-Eq4} multiplied by $V_1^{2^s}$ and \eqref{main-Eq5} multiplied by $U_1$ gives
	$$ \left( U_2V_1^{2^s} + V_2^{2^s}U_1 \right)z^{2^{m+s}} + \left(U_3V_1^{2^s} + V_3^{2^s}U_1  \right)z^{2^s} =0, $$
	i.e.,
	$$\left( U_2V_1^{2^s} + V_2^{2^s}U_1 \right) \left(z^{2^{m+s}}+z^{2^s}\right)=0$$
	since $U_2V_1^{2^s} + V_2^{2^s}U_1+ U_3V_1^{2^s} + V_3^{2^s}U_1=0$
	 by Lemma \ref{key-lemma} (iv). Moreover, by Lemma \ref{key-lemma} (v), $U_2V_1^{2^s} + V_2^{2^s}U_1\neq0$ and thus $z^{2^{m+s}}+z^{2^s}=0$. In other words, $z\in\gf_{2^{m}}$. Plugging it into \eqref{main-Eq}, we get 
	$$(C+E)(z^{2^s}+z)=0.$$
	By Lemma \ref{key-lemma} (i), $C+E\neq0$ and thus $z^{2^s}+z=0$. Then $z\in\gf_{2^{\gcd(m,s)}}=\gf_{2}$. 
	
	In conclusion, the equation \eqref{main-Eq} has only two solutions $z = 0, 1$ in $\gf_{2^{3m}}$ for any $a\in\gf_{2^{3m}}^{*}$ and then $f$ is APN over $\gf_{2^{3m}}$.
\end{proof}

\section{Conclusion and further work}

In this paper, we obtained two new infinite classes of APN functions over $\gf_{{2^{2m}}}$ and $\gf_{{2^{3m}}}$. The first one is with bivariate form.  Moreover, we showed that our APN families are CCZ-inequivalent to all known infinite families of APN functions by their $\Gamma$-ranks over $\gf_{{2^8}}$ or $\gf_{{2^9}}$. Furthermore, it is worth mentioning that one of our APN families covers an APN function over $\gf_{{2^8}}$ found by Edel and Pott in \cite{EdelP09} using the switching method. 

Notice that the newly found APN family over $\gf_{{2^m}}^2$ belongs to the univariate form  over $\gf_{q^2}$ with $q=2^m$
$$f(z) = z^3+Az^{3q}+Bz^{2q+1}+Cz^{q+2}+Dz^{5}+Ez^{5q}+Fz^{4q+1}+Gz^{q+4}+Hz^{q+1}+Iz^{2(q+1)}.$$
Thus the next question is whether it is possible to  obtain more new APN functions over $\gf_{q^2}$ from functions of the above form, or more generally, 
\begin{eqnarray*}
	f(z) &=& z(Az^2+Bz^4+Cz^q+Dz^{2q}+Ez^{4q}) + z^2(Gz^4+Hz^q+Iz^{2q}+Jz^{4q})\\
	&& + z^4(Kz^q+Lz^{2q}+Mz^{4q}) + z^q (Nz^{2q}+Pz^{4q}) +Qz^{6q},
\end{eqnarray*}
which has been preliminarily discussed in \cite{Budaghyan2008}. Moreover, for the newly found APN family over $\gf_{{2^{3m}}}$, there are two important problems worth studying: (1) We checked with a personal computer that for $m\le 8$, there always exist elements $\mu$ and $s$ such that $\mu^{2^{2m}+2^m+1}\neq1$ and $L(z) = z^{2^{m+s}}+\mu z^{2^s} + z$ permutes $\gf_{{2^{3m}}}$. Can a theoretical proof of the existence of such $\mu$ and $s$ be given? (2) Study the CCZ-equivalence among the APN functions $f(z)=(z^{2^{m+s}}+\mu z^{2^s} + z)^{2^m+1}+vz^{2^m+1}$ with different parameters $\mu,s,v$ and determine a lower bound about the number of CCZ-inequivalent APN functions over $\gf_{{2^{3m}}}$ of the above form.

\label{conclusion}
\bibliographystyle{plain}

\bibliography{ref}

\end{document}